\documentclass[letterpaper, 10 pt, conference]{ieeeconf}  

\let\labelindent\relax

\usepackage{graphicx}
\usepackage{epsfig}
\usepackage{mathptmx}
\usepackage{times}
\usepackage{amsmath} 
\usepackage{amsthm}
\usepackage{amssymb}
\usepackage{tcolorbox}
\usepackage{cite}
\usepackage[ruled,vlined]{algorithm2e}
\usepackage[colorlinks = true, linkcolor = blue, citecolor =
blue]{hyperref}
\usepackage{enumitem}
\usepackage{subcaption}
\usepackage[mathscr]{euscript}
\usepackage{xcolor}

\usepackage{mathtools}
\usepackage[export]{adjustbox}

\usepackage{multicol}

\newcommand\new[1]{{\color{black}#1}}

\newtheorem{thm}{Theorem}
\newtheorem{prop}[thm]{Proposition}
\newtheorem{lem}[thm]{Lemma}

\newtheorem{rem}[thm]{Remark}

\newcommand{\ldef}{:=}

\newcommand{\Mc}[1]{\mathcal{#1}}
\newcommand{\real}{\ensuremath{\mathbb{R}}}

\newcommand{\nat}{{\mathbb{N}}}
\newcommand{\natz}{{\mathbb{N}}_0}

\newcommand{\dd}{\mathrm{d}}
\newcommand{\sinc}{\text{sinc}}
\newcommand{\norm}[1]{\left\lVert #1 \right\rVert}

\DeclareMathOperator*{\argmin}{arg\,min}

\newcommand{\param}{\textbf{a}}
\newcommand{\tth}{^\text{th}}
\newcommand{\sys}{\mathscr{S}}

\newcommand{\thmtitle}[1]{\mbox{}\textit{(#1).}}

\renewcommand{\theenumi}{(\arabic{enumi})}

\newcommand{\remend}{\relax\ifmmode\else\unskip\hfill\fi\hbox{$\bullet$}}

\title{Event-Triggered Polynomial Control for Trajectory Tracking by 
Unicycle Robots}

\author{Harini V$^{1}$, Anusree Rajan$^{\new{1}}$, Bharadwaj 
	Amrutur$^{1}$, and Pavankumar Tallapragada$^{\new{1}}$ 
		\thanks{$^{1}$ Robert Bosch Centre for Cyber Physical Systems, Indian 
		Institute of Science, Bengaluru. {\tt\small \{hariniv, anusreerajan, amrutur, 
			pavant\}@iisc.ac.in}}%
}

\begin{document}

		\maketitle
		\thispagestyle{empty}
		\pagestyle{empty}

	\begin{abstract}
	This paper proposes an event-triggered polynomial control method for 
	trajectory tracking by unicycle robots. In this method, each 
	control input between two consecutive events is a 	
	polynomial \new{and its coefficients} are chosen to minimize the error in approximating a 
	continuous-time control signal. We design an event-triggering rule 
	that guarantees uniform ultimate boundedness of the tracking error and \new{non-Zeno behavior of inter-event times}. We illustrate 
	our results through a suite of numerical simulations and experiments, \new{which indicate} that the number of 
	events generated by the proposed controller is significantly less 
	compared to that by a time-triggered controller or a event-triggered 
	controller based on zero-order hold while guaranteeing similar 
	tracking performance.	
	\end{abstract}

\section{INTRODUCTION}

\subsection{Motivation}
Trajectory tracking for mobile robots is a well-studied problem 
with many applications, such as industrial automation, military 
surveillance and multi-robot coordination. An important challenge
in these applications may be constrained resources, such as 
communication, computation, and energy. Event-triggering~\cite{PT:2007, 
WH:2012, ML:2010} is a popular method for control under 
such resource constraints. In the event-triggered control (ETC) literature, 
design of control laws is mostly based on the zero-order-hold (ZOH) 
technique. However, for many of the communication protocols used in control applications, such as 
TCP and UDP~\cite{DH:2020}, there is a minimum packet size. ZOH control may therefore result in under 
utilization of each packet while simultaneously increasing the number 
of communication instances. \new{Whereas}, non-ZOH 
control could improve the utilization of each packet while reducing the 
number of communication instances. 
\new{With this motivation}, in this paper, we propose an event-triggered 
polynomial control method for trajectory tracking of unicycle robots.

\subsection{Literature Review}

The literature on trajectory tracking with event-triggered
communication includes~\cite{PT-NC:2013}, which designs 
an event-triggered tracking controller for non-linear systems that 
guarantees uniform ultimate (UU)
boundedness of the tracking error and non-Zeno behavior of inter-event 
times (IETs). Similarly, references~\cite{RP-etal:2015, CX-YF-JQ:2020, 
santos2020lyapunov} propose Lyapunov based event or self-triggered 
tracking controllers for mobile robots by emulating a continuous time 
controller and guarantee ultimate boundedness of the tracking error. 
Reference~\cite{santos2014adaptive} 
obtains a linear system model for a 
Pioneer robot by system identification and designs an adaptive 
self-triggered tracking controller. Reference~\cite{KV-etal:2017} 
proposes an event-triggered optimal tracking control method for 
nonlinear systems using ideas from reinforcement learning.
Whereas, reference~\cite{QC-etal:2019} presents a self-triggered model predictive control (MPC) strategy for trajectory tracking of unicycle-type robots with input constraints and bounded disturbances.
Reference~\cite{PG-etal:2022} deals with the tracking control of 
quadrotors with external disturbances and proposes an event-triggered 
sliding mode control strategy. The recent paper~\cite{PZ-etal:2023} 
proposes both event-triggered and self-triggered saturated feedback 
control strategies for trajectory tracking of unicycle mobile robots.
In all these papers, \new{except in~\cite{QC-etal:2019}}, the control input to the plant is held constant between two successive events.

\new{Model-based ETC~\cite{EG-PA:2013, MH-FD:2013} is one among the limited number of works on non-ZOH ETC}. In 
this method the control input to the plant is time-varying even between 
two successive events and is generated using a model of the plant whose
state is updated in an event-triggered manner. 
\new{Other control methods based on non-ZOH control are event/self-triggered MPC~\cite{HL-YS:2014, HL-etal:2018} and event-triggered dead-beat 
control (DBC)~\cite{BD-etal:2017}. In these methods, 
the controller transmits a sequence of control inputs to the actuator 
at each triggering instant and the actuator applies this control sequence to the plant 
until the next control packet arrives. 
References~\cite{AL-JS:2023,KH-etal:2017} extend this idea based on first-order-hold (FOH), where, the control input  is linearly interpolated between sampling points in the prediction horizon.}

\new{Our recent work~\cite{AR-PT:2023} proposes a novel non-ZOH based ETC method, called as event-triggered parameterized control (ETPC) method, 
for stabilization of linear systems. In~\cite{AR-PT:2024}, we extend this control method to nonlinear control settings with external disturbances. There are also a few papers~\cite{AR-AK-PT:2024,SD-eal:2023} that use a parameterized control law in MPC like problems but not with event-triggering.}

\subsection{Contributions}
The major contributions of this work are \new{as follows}:
\begin{itemize}
	\item We propose an event-triggered polynomial control (ETPC) method for 
	trajectory tracking of unicyle robots. The proposed controller guarantees UU boundedness of the tracking 
	error and non-Zeno behavior of IETs.
	\item We present the results of practical implementation on a ground 
	robot. This is a contribution to the as yet limited literature on 
	practical implementations of event-triggered controllers.
	\item Through simulations and experiments, we show significant 
	reduction in communication compared to the standard ZOH
	based ETC.
	\item \new{Compared to other existing methods, such as model based ETC, event/self-triggered MPC and event-triggered DBC, the proposed method allows for applying a time-varying control input to the plant even between two successive events using limited computational resources at the actuator and transmitting limited information over the communication network.}
	\item \new{The only papers that consider an event-triggered control method similar to 
	the one proposed is~\cite{AR-PT:2023,AR-PT:2024}, where the objective is stabilization. On the other hand, in the current paper, the context is 
	trajectory tracking for unicycle robot models. 
The results in~\cite{AR-PT:2024}, which considers a time invariant nonlinear system, are not directly applicable here for the stability analysis of tracking error as the tracking error dynamics explicitly depends on time $t$. Moreover, in this paper, we validate our results practically through several experiments.
}
\end{itemize}

\subsection{Notation}
Let $\real$ denote the set of all real numbers. Let $\nat$ and $\natz$ 
denote the set of all positive and non-negative integers, 
respectively.  For any $x \in \real^n$, $\norm{x}$ denotes the 
euclidean norm of $x$. 
For any right continuous function $f: \real_{\ge 
	0} \to \real^n$ and $t \ge 0$, $f(t^{+}) \ldef \lim\limits_{s \to 
	t^{+}} f(s)$. For any two functions $v, w : [0, T] \rightarrow \real$, let
\begin{equation*}
\langle v, w \rangle := \int_{0}^{T} v( \tau ) w( \tau ) \dd \tau.
\end{equation*}

\section{PROBLEM SETUP}	\label{sec:problem_setup}

In this paper, we propose a tracking control 
method that works best in cases where communication is significantly 
more costly than computation. In this section, we present the system 
model and the objective of this paper.
\subsection{System Dynamics}
Consider the unicycle model of a robot,
\begin{equation}\label{eq:unicycle_robot}
\dot{x}=v\cos \theta, \quad
\dot{y}= v\sin \theta, \quad
\dot{\theta} = \omega,
\end{equation}
where $(x,y)$ denotes the position of the robot and $\theta$ denotes 
the orientation of the robot, which is the angle between the heading 
direction of the robot and the $x-$axis. 
$v$ and 
$\omega$ denote, respectively, the linear velocity and the angular 
velocity of the robot, which are the control inputs. 
The position and 
the orientation of the robot are continuously available to the 
controller.
The robot has to track a given reference trajectory which satisfies the 
following model,
\begin{equation}\label{eq:reference_trajectory}
\dot{x}_r=v_r \cos \theta_r, \quad
\dot{y}_r= v_r \sin \theta_r, \quad
\dot{\theta}_r = \omega_r,
\end{equation}
where $(x_r,y_r)$ and $\theta_r$ denote, respectively, the reference
position and the reference orientation. $v_r$ and $\omega_r$ denote the 
inputs to the reference system. We make the following assumption on the 
reference inputs, just as in~\cite{RP-etal:2015}.

\begin{enumerate}[resume, label=\textbf{(A\arabic*)}]
	\item There exists $M \ge 0$ such that $|v_r(t)|, |\omega_r(t)|,
	|\dot{v}_r(t)|$ and $|\dot{\omega}_r(t)|$ are upper
	bounded by $M, \ \forall t\ge t_0$. Moreover, there exists a $c >0$
	such that $|v_r(t)| \ge c, \ \forall t \ge t_0$.  \label{A:A1}  
	\remend
\end{enumerate}
 We also assume that the reference trajectory and the reference inputs 
 are available to the controller a priori. 

The tracking error, in the robot frame, can be represented as follows,
\begin{equation}\label{eq:tracking_error}
\begin{bmatrix}
x_e \\ y_e \\ \theta_e
\end{bmatrix}
:=\begin{bmatrix}
\cos\theta && \sin\theta && 0 \\
-\sin\theta && \cos\theta && 0 \\
0 && 0 && 1
\end{bmatrix}
\begin{bmatrix}
x_r - x \\ y_r - y \\ \theta_r - \theta
\end{bmatrix}.
\end{equation}
The evolution of tracking error in robot frame is given by
\begin{equation}\label{eq:control_affine_form}
\dot{X}=F(X,t)+G(X)u,
\end{equation}
where $X = \begin{bmatrix}
x_e & y_e & \theta_e
\end{bmatrix}^{\top}$, $u = \begin{bmatrix}
v & \omega
\end{bmatrix}^{\top}$,
\begin{equation*}
F(X,t)=\begin{bmatrix}
v_r\cos \theta_e & v_r\sin\theta_e & \omega_r 
\end{bmatrix}^{\top},
\end{equation*}
\begin{equation*}
G(X) = \begin{bmatrix}
-1 & 0 & 0 \\ y_e & -x_e & -1
\end{bmatrix}^{\top}.
\end{equation*}

In this paper, we wish to design a controller for
trajectory tracking by unicycle robots, \new{where the controller communicates with the actuator over a communication network in an event-triggered manner}. We wish to
apply a time varying control input to the plant even between two
successive communication times \new{while}
transmitting limted information at each communication instance. So, we consider a polynomial control input whose coefficients are updated \new{at each event}. In particular, each control input to the 
plant is a polynomial of degree $p$ between any two 
consecutive events. Now, let $u_1(t)\ldef v(t)$ and $u_2(t) \ldef 
\omega(t)$. Then, for $i \in \{1,2\}$ $\forall \tau \in [0, t_{k+1} - t_k)$,
\begin{equation}\label{eq:control_law}
u_i(t_k + \tau) = f(\param_i(k),\tau) \ldef \sum_{j=0}^{p} 
a_{ji}(k)\tau^j.
\end{equation}
Here $(t_k)_{k \in \natz}$ is the sequence of communication time 
instants from the controller to the actuator. At $t_k$, the controller 
communicates the coefficients of the polynomial control input, 
$\param(k) \ldef [a_{ji}(k)]\in \real^{(p+1) \times 2}$, to the 
actuator. We also let $\param_i(k)$ denote the $i\tth$ column of 
$\param(k)$. \new{Note that, the results in this paper are easily extendable to the more generalized parameterized control law proposed in~\cite{AR-PT:2023}.}

Note also that, in contrast to the usual trend in 
ETC 
literature, here the control input to the robot is not held constant 
between two communication time instants. In this control method, at each communication time instant, the controller has to send only the coefficients of the polynomial control input and the actuator can easily generate the time varying control input to the plant once it receives the coefficients from the controller.

Figure~\ref{fig:ETPC_system} depicts the general configuration of the ETPC system considered in this paper.
\begin{figure}[h]
	\centering
	\includegraphics[width=7cm]{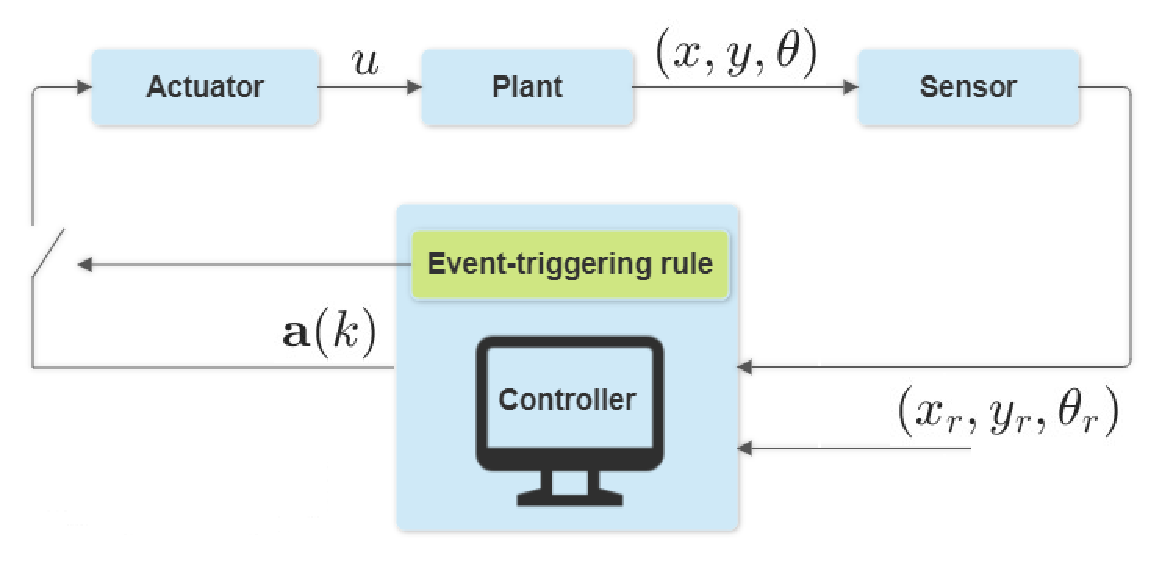}
	\caption{Event-triggered polynomial control configuration}
	\label{fig:ETPC_system}
\end{figure} 
Here, the controller has continuous access to the plant state $(x,y,\theta)$ and the reference trajectory $(x_r,y_r,\theta_r)$. We also assume that the controller has enough computational resources to evaluate the 
event-triggering condition and to update the coefficients of the 
polynomial control input at an event-triggering instant.

\subsection{Objective}

Our objective is to design a polynomial control law and an 
event-triggering rule to implicitly determine the communication time 
instants at which the coefficients of the 
polynomial control input are updated, so that the tracking error is 
uniformly ultimately bounded. We also wish to ensure the absence of Zeno behavior.

\section{DESIGN OF THE ETPC SYSTEM}\label{sec:design}
In this section, we design a polynomial control law as well as an event-triggering rule to achieve our objective.

\subsection{Control Law}
We first consider the continuous time feedback control signal, 
$\hat{u} \ldef \begin{bmatrix}
  \hat{v} & \hat{\omega}
\end{bmatrix}^T$, that was proposed in~\cite{ZJ-HN:1997}. In particular,
\begin{equation}\label{eq:jiang_controller}
\begin{aligned}
&\hat{v}(\hat{X},t) := v_1(\hat{X},t)+c_1(\hat{x}_e-c_3\hat{\omega}(\hat{X},t)\hat{y}_e), \\ &\hat{\omega}(\hat{X},t) := \omega_r + \gamma \hat{y}_e v_r \sinc \hat{\theta}_e + c_2 \gamma \hat{\theta}_e,
\end{aligned}
\end{equation}
where
\begin{equation}
\begin{aligned}\label{eq:x_1_hat}
& 
\begin{aligned}
v_1(\hat{X},t) := & v_r \cos \hat{\theta}_e + c_3 \hat{\omega}(\hat{X},t)(\hat{\omega}(\hat{X},t) \hat{x}_e - v_r \sin \hat{\theta}_e) \\ &- c_3 v_2(\hat{X},t) \hat{y}_e, 
\end{aligned}
\\
 & 
\begin{aligned}
v_2(\hat{X},t) := &\dot{\omega}_r+\gamma v_r \sinc \hat{\theta}_e(-\hat{\omega}(\hat{X},t)\hat{x}_e+v_r\sin\hat{\theta}_e) + \\ & \gamma \hat{y}_e \dot{v}_r\sinc\hat{\theta}_e
+(\gamma \hat{y}_e v_r \sinc'\hat{\theta}_e+c_2\gamma)(\omega_r-\hat{\omega}(\hat{X},t)).
\end{aligned}
\end{aligned} 
\end{equation}
Here $\hat{X} \ldef \begin{bmatrix}
\hat{x}_e & \hat{y}_e & \hat{\theta}_e
\end{bmatrix}^T$ evolves as,
\begin{equation}\label{eq:simulated_sys}
\dot{\hat{X}}=F(\hat{X},t)+G(\hat{X})\hat{u}(\hat{X},t), \quad \forall t \in [t_k,t_{k+1}),
\end{equation}
where $\hat{X}(t_k)=X(t_k)$ for all $k \in \natz$. Note also that 
$\sinc'\hat{\theta}_e$ denotes the derivative of $\sinc\hat{\theta}_e$ 
with respect 
to $\hat{\theta}_e$ and $c_1, c_2, c_3, \gamma > 0$ are design 
parameters. Reference~\cite{ZJ-HN:1997} considers the tracking error 
evolution~\eqref{eq:control_affine_form} with $u=\hat{u}(X,t)$ and 
shows global convergence of tracking error to zero under some 
conditions on the reference inputs.

Our idea is to find the best polynomial approximation of the control signal~\eqref{eq:jiang_controller}. At each communication time instant $t_k$, the coefficients of the polynomial control input~\eqref{eq:control_law} are updated by solving the following finite horizon optimization problems, for $i \in \{1,2\}$,

\begin{equation}\label{eq:a_k}
\begin{aligned}
\param_i(k) \in \argmin_{a \in \real^{p+1}} & 
\int_{0}^{T} \left| 
f(a,\tau)-\hat{u}_i(\hat{X},t_k+\tau)\right|^2+\delta_i\left|f(a,\tau)\right|^2
 \dd \tau,\\
\textrm{s.t.} \quad &  f(a,0)=\hat{u}_i(\hat{X}(t_k),t_k),
\end{aligned}
\end{equation} 
where $\hat{u}_1 \ldef \hat{v}$ and $\hat{u}_2 \ldef \hat{\omega}$. 
Here, $T>0$ is a finite time horizon which is to be designed and 
$\delta_{1}, \delta_{2} \ge 0$ are design parameters which are useful 
for penalizing large magnitudes of the control input signal.
Note that, in order to solve the optimization problem~\eqref{eq:a_k}, we require the values of $\hat{v}$ and $\hat{\omega}$ over the time horizon $(t_k,t_k+T]$. These values are estimated by numerically simulating the system~\eqref{eq:simulated_sys}.

Note that, the only constraint in optimization problem~\eqref{eq:a_k} 
fixes the value of $a_0$. Hence, letting
\begin{equation*}
  \bar{u}_i(\tau) \ldef \hat{u}_i(\hat{X},t_k+\tau), \quad \mu_i := 
  \bar{u}_i(0),
\end{equation*}
we can rewrite~\eqref{eq:a_k} as the following unconstrained 
optimization problem, for $i \in \{1,2\}$,
\begin{equation}\label{eq:a_ik}
\bar{\param}_i(k) \in \argmin_{a \in \real^p} J_i(a), 
\end{equation} 
where,
\begin{equation*}
\begin{aligned}
J_i(a)=	&\langle\bar{u}_i,\bar{u}_i\rangle +(1+\delta_{i}) \mu_i^2 T - 
2 \mu_i \langle\bar{u}_i,1\rangle - 2 
\sum_{j=1}^{p}a_j\langle\bar{u}_i,\tau^j\rangle\\ & 
+(1+\delta_{i})\sum_{j=1}^{p} \left[ 
\sum_{l=1}^{p}a_ja_l\frac{T^{j+l+1}}{j+l+1} + 2 \mu_i 
a_j\frac{T^{j+1}}{j+1} \right] .
\end{aligned}
\end{equation*}
Thus, we have
\begin{equation*}
\param_i(k)=\begin{bmatrix}
\mu_i & \bar{\param}^T_i(k)
\end{bmatrix}^T.
\end{equation*} 
Viewed this way, we see that Problem~\eqref{eq:a_k} is always 
feasible.
Note that, given the coefficients of the polynomial control input that are obtained by 
	solving~\eqref{eq:a_ik}, the control input that is applied by the 
	actuator is as given in~\eqref{eq:control_law}.
\begin{rem}\thmtitle{Control signal for $\tau > T$}
	\new{As discussed in Remark 1 of~\cite{AR-PT:2024}, control input $u(t_k + \tau)$ is well defined for all $\tau \in [t_k, t_{k+1})$ even if 
	$t_{k+1} - t_k > T$. }
	\remend
\end{rem}

\begin{prop}\label{prop:convexity}
	The optimization problem~\eqref{eq:a_ik} is a strictly convex optimization problem.
\end{prop}
\begin{proof}
	Proof of this result follows along similar lines as in the proof of Proposition 2 in~\cite{AR-PT:2024}.
	For all $i \in \{1,2\}$, the Hessian matrix of $J_i(.)$, denoted as 
	$\textbf{H}_i$, 
	is as follows,
	\begin{equation*}
	\textbf{H}_{i}=2(1+\delta_i)
	\begin{bmatrix}
	\frac{T^3}{3} & 	\frac{T^4}{4} & \ldots 
	& \frac{T^{p+2}}{p+2}\\
	\frac{T^{4}}{4} & 	\frac{T^{5}}{5} & \ldots 
	& \frac{T^{p+3}}{p+3}\\
	\ldots & 	\ldots & \ldots & \ldots\\
	\frac{T^{p+2}}{p+2} & 	\frac{T^{p+3}}{p+3} & \ldots 
	& \frac{T^{2p+1}}{2p+1}
	\end{bmatrix}, \quad \forall i \in \{1,2\}.
	\end{equation*}
	Observe that $\textbf{H}_{i}$ 
	is $2(1+\delta_i)$ times the Gram 
	matrix for the functions in $\{\tau^j: [0,T] \to \real\}_{j=1}^p$, 
	which is a set of linearly independent functions. Thus, 
	$\textbf{H}_i$ is a positive definite matrix $\forall i \in \{1,2\}$. Hence, the cost 
	function in~\eqref{eq:a_ik} is strictly convex. As there are no constraints 
	in~\eqref{eq:a_ik}, it is a strictly convex optimization problem.
\end{proof}

	Note also that the computational requirement of the proposed 
	controller is similar to that of the existing control methods like 
	the event-triggered MPC or the event-triggered 
	DBC.

\subsection{Event-Triggering Rule}

We consider the candidate Lyapunov function,
\begin{equation}\label{eq:Lyap_fn}
V(X,t)= \frac{1}{2}x_1^2+\frac{1}{2}y_e^2+\frac{1}{2\gamma}\theta_e^2,
\end{equation}
where $x_1 := x_e-c_3\hat{\omega}(X,t)y_e$, to design the event-triggering 
rule. Please note that in $x_1$, it is indeed $\hat{\omega}(X,t)$ and 
not $\hat{\omega}(\hat{X},t)$. Letting $ e(t) \ldef u(t)-\hat{u}(X,t) ,$
we see that the derivative of $V$ along the trajectories of the sampled 
data system~\eqref{eq:control_affine_form}-\eqref{eq:control_law} can 
be expressed as
	\begin{equation}\label{eq:V_dot}
\begin{aligned}
\dot{V}&=\frac{\partial V}{\partial t}+\frac{\partial V}{\partial X} \dot{X}=\frac{\partial V}{\partial t}+\frac{\partial V}{\partial X} (F(X,t)+G(X)u)
\\&= -\Sigma(X,t)+\Lambda(X,e,t)
\end{aligned}
\end{equation}
where
\begin{equation}\label{eq:Sigma}
\Sigma(X,t) \ldef c_1x_1^2+c_2\theta_e^2+c_3\hat{\omega}^2(X,t)y_e^2,
\end{equation}
\begin{equation}\label{eq:Lambda}
\Lambda(X,e,t) \ldef \frac{\partial V}{\partial X}G(X)e(t) .
\end{equation}
Equation~\eqref{eq:V_dot} follows directly from equation 
(41) in~\cite{ZJ-HN:1997}. 

Now, we define the event-triggering rule as follows,
\begin{equation}\label{eq:etr}
t_{k+1}\ldef \min\{t>t_k: \dot{V} \ge -\sigma \Sigma(X,t) \; \text{and} \; V(X,t) \ge \epsilon^2\},
\end{equation}
where $t_0 \ldef 0$ and $\sigma \in (0,1)$, $\epsilon^2>0$ are design 
parameters.

In summary, the complete system, $\sys$, is the combination 
of the reference system~\eqref{eq:reference_trajectory}, the unicycle 
robot model~\eqref{eq:control_affine_form}, the polynomial 
control law~\eqref{eq:control_law}, with coefficients chosen by 
solving~\eqref{eq:a_ik}, which are updated at the events determined by 
the event-triggering rule~\eqref{eq:etr}. That is,
\begin{equation}\label{eq:full-system}
  \sys : \ \eqref{eq:reference_trajectory}, 
  \eqref{eq:control_affine_form}, \eqref{eq:control_law}, 
  \eqref{eq:a_ik}, \eqref{eq:etr}.
\end{equation}

\section{ANALYSIS OF THE ETPC SYSTEM}\label{sec:analysis}
In this section, we analyze the performance of the designed 
ETPC system. We show that for system~\eqref{eq:full-system}, the tracking error is uniformly 
ultimately bounded and the IETs have a uniform positive 
lower bound. We first present a couple of lemmas that help to prove the main 
result of this paper.
\begin{lem}\label{lem:V}
  Consider system~\eqref{eq:full-system} and Lyapunov 
  function~\eqref{eq:Lyap_fn}. Let Assumption~\ref{A:A1} hold and 
  $\epsilon_k^2 \ldef V(X(t_k),t_k)$.
  Then, for any $c_1, c_2, 
  c_3, \epsilon^2>0$ and $\gamma > 0$ sufficiently large, $V(X(t),t)  \le \epsilon_k^2, \forall t \in [t_k,t_{k+1})$ and 
  $\forall k \in \nat$.
\end{lem}
\begin{proof}  
  Note that, as $e(t_k^+)=0$ for any $k \in \natz$, $\dot{V}(X(t_k^+),t_k^+)=-\Sigma(X (t_k^+),t_k^+)$. Further the event-triggering rule~\eqref{eq:etr} implies that $\epsilon_k^2 \ge \epsilon^2$ for all $k \in \nat$ and hence $\dot{V}(X(t_k^+),t_k^+) < 0$.  The last inequality follows from Lemma $4$ in~\cite{RP-etal:2015}, which states that there exists $\mathscr{V} >0$ for any $c_1, c_2, c_3, \epsilon^2>0$ and $\gamma > 0$ sufficiently large such that $V(X,t)\ge \epsilon^2$ implies $\Sigma(X,t) \ge \mathscr{V} >0$. 
  Now, let us prove the statement that $V(X(t),t) \le \epsilon_k^2, \forall t \in [t_k,t_{k+1})$ and $\forall k \in \nat$ by contradiction. Suppose that this statement is not true. Then, as $V(X,t)$ is a differentiable function of time, there must exist $\bar{t} \in (t_k,t_{k+1})$, for some $k \in \nat$, such that $V(X(\bar{t}),\bar{t})=\epsilon^2_k$ and $\dot{V}(X(\bar{t}),\bar{t})>0$. However, as $\bar{t}<t_{k+1}$, we can say that the triggering condition is not satisfied at $t=\bar{t}$ and hence $\dot{V}(X(\bar{t}),\bar{t}) < -\sigma \Sigma(X(\bar{t}),\bar{t}) < 0$. As there is a contradiction, we conclude that there does not exist such a $\bar{t}$ and hence the result is true.
\end{proof}
\begin{rem}\label{rem:V}
  Under Assumption~\ref{A:A1}, $V(X,t)$ is a continuously differentiable positive definite radially unbounded function of $X$. \remend
\end{rem}
Next, we show that the control $u$ and its time 
derivative are uniformly bounded between any two consecutive events. 
This result helps us to prove that the IETs generated by 
the proposed ETPC method do not exhibit Zeno behavior. 
\begin{lem}\label{lem:claim}
	Consider system~\eqref{eq:full-system} and Lyapunov 
	function~\eqref{eq:Lyap_fn}. Let Assumption~\ref{A:A1} hold and 
	$\epsilon_k^2 \ldef V(X(t_k),t_k)$.
	Then, there exist monotonically increasing functions $\beta_1 : \real_{> 0} \to \real_{> 
		0}$, $\beta_2: \real_{> 0} \to \real_{> 0}$ such that $\norm{u(t)} 
	\le \beta_1(\epsilon_k^2)$ and $\norm{\dot{u}(t)} \le 
	\beta_2(\epsilon_k^2)$, $\forall t \in [t_k, \min\{ t_{k+1}, t_k+T 
	\})$, $\forall k\in \nat$.
\end{lem}

\begin{proof}
	Note that, $\forall i \in 
	\{1,2\}$ and for any $k \in \nat$, $u_i(t)$ for $t \in 
	[t_k, t_{k+1})$ is chosen by solving the unconstrained strictly convex optimization 
	problem~\eqref{eq:a_ik}. 
	Thus, the stationarity condition $\frac{\partial}{\partial 
		a}J_i(a)=0$ is necessary and sufficient for $a$ to be the optimizer of 
	problem~\eqref{eq:a_ik} for $i \in \{1,2\}$. As a result, 
	the optimizers of problem~\eqref{eq:a_ik} are the solutions of the equation $\textbf{H}_{i} \bar{\param}_i(k) = D_i(k)$,
	where $\textbf{H}_{i}$ is the Hessian matrix \new{of $J_i(.)$} and
	\begin{align*}
		D_i(k)= & 2\begin{bmatrix}
	\langle\bar{u}_i,\tau^1\rangle &
	\langle\bar{u}_i,\tau^2\rangle &
	\ldots &
	\langle\bar{u}_i,\tau^p\rangle
	\end{bmatrix}^{\top} \\ & - 2(1+\delta_{i}) \mu_i \begin{bmatrix}
	\frac{T^2}{2}&\frac{T^3}{3} &\ldots&\frac{T^{p+1}}{p+1} 
	\end{bmatrix}^{\top}.
	\end{align*}
	\new{As $\textbf{H}_{i}$ 
	is invertible,} there is a unique optimal solution 
	$\bar{\param}_i(k)$ to the problem~\eqref{eq:a_ik} and is equal to
	$\bar{\param}_i(k) = \textbf{H}_{i}^{-1} D_i(k)$. 
	
	Now, note that, $V(\hat{X}(t),t)\le V(\hat{X}(t_k),t_k)$ for all $t 
	\in [t_k,\min\{t_{k+1},t_k+T\})$ and for any $k \in \nat$ as 
	$\dot{V}(\hat{X},t)=-\Sigma(\hat{X},t)\le 0$ where $\hat{X}$ evolves 
	as~\eqref{eq:simulated_sys}. As $\hat{X}(t_k)=X(t_k)$, $ 
	V(\hat{X}(t_k),t_k)=V(X(t_k),t_k)$ for each $k \in \nat$. According 
	to Remark~\ref{rem:V}, we can say 
	that there exists a class $\Mc{K}$ function $\alpha^{\prime}(.)>0$ 
	such that $\norm{\hat{X}(t)} \le \alpha^{\prime}(\epsilon_k^2)$ for 
	all $t \in [t_k,\min\{t_{k+1},t_k+T\}]$ for each $k \in \nat$. This 
	implies that, for each $i \in \{1,2\}$, for all $\tau \in 
	[0,\min\{t_{k+1}-t_k,T\})$ and $\forall k \in \nat$, 
	$|\bar{u}_i(\tau)|$ is upper bounded by a monotonically increasing 
	positive real valued function of $\epsilon_k^2$. By using this fact, we can say that there exists a monotonically increasing 
	function $\beta^{\prime}(.)$ such that $\norm{\param(k)}\le 
	\beta^{\prime}(\epsilon_k^2)$, $\forall k \in \nat$. Thus, we can say 
	that there exists monotonically increasing functions $\beta_1, \beta_2 : \real_{> 0} \to \real_{> 0}$ 
	such that $\forall t \in [t_k, \min\{ t_{k+1}, t_k+T \})$, $\forall 
	k\in \nat$, $\norm{u(t)} \le \norm{\param(k)} \norm{\begin{bmatrix}
		1 & t - t_k & \ldots & (t - t_k)^p
		\end{bmatrix}^{\top}} \le \beta_1(\epsilon_k^2)$ and $	\norm{\dot{u}(t)} \le \norm{\param(k)} \norm{ \begin{bmatrix}
			0 & 1 & \ldots & p(t - t_k)^{p-1} \end{bmatrix}^{\top}} \le \beta_2(\epsilon_k^2).$
	This proves Lemma~\ref{lem:claim}.
\end{proof}

Next, we present the main theorem of this paper that shows that the IETs do not exhibit Zeno behavior and the tracking error is uniformly ultimately bounded.

\begin{thm}\thmtitle{Absence of Zeno behavior and UU boundedness of tracking error}\label{thm:non-Zeno}
	Consider system~\eqref{eq:full-system}. Suppose Assumption~\ref{A:A1} 
	holds. Then, 
	\begin{itemize}
		\item the IETs, $t_{k+1} - t_k$ for $k \in \nat$, 
		are uniformly lower bounded by a positive real number that depends 
		on the bound of the initial tracking error.
		\item moreover, the lower bound on 
		the IETs converges to a positive real number, which is 
		independent of the initial tracking error, in finite time.
		\item the tracking error is uniformly ultimately bounded.
	\end{itemize}
	
\end{thm}

\begin{proof}

	We first prove the first statement of this theorem. Note that, 
	for each $k \in \nat$, $\dot{V}(X(t_k^+),t_k^+) = 
	-\Sigma(X(t_k),t_k)$ as $e(t_k^+)=0$. Hence, according to 
	the event-triggering rule~\eqref{eq:etr} for each $k \in \nat$, the 
	inter-event 
	time $t_{k+1}-t_k$ must be greater than the time it takes 
	$\Lambda(X,e,t)$ to grow from $0$ to $(1-\sigma)\Sigma(X,t)$. 

  Let $\epsilon_k^2 \ldef V(X(t_k),t_k)$. Given Lemma~\ref{lem:V} 
  and~\eqref{eq:Lambda}, we can find an upper bound on $\Lambda(X,e,t)$ 
  as follows,
	\begin{align*}
	\Lambda(X,e,t) & \le L(\epsilon_k^2)\norm{e}, \quad \forall t \in [t_k,t_{k+1}),
	\end{align*}
	where $L : \real_{> 0} \to \real_{> 0}$ is a continuous function defined as,
	\begin{equation*}
	L(R) \ge \max_{V \le R}\norm{\frac{\partial V}{\partial X}G(X)}.
	\end{equation*}
	According to Remark~\ref{rem:V}, the right hand side of the above 
	inequality exists as any sub-level set of $V$ is compact.
  Also note that the event-triggering rule~\eqref{eq:etr} implies that 
  for each $k \in \nat$, $V(X(t_k), t_k) \geq \epsilon^2$. Hence, we 
  can say that for any $k \in \nat$, the IET $t_{k+1}-t_k$ 
  is lower bounded by the time it takes $\norm{e}$ to grow from $0$ to 
  $(1-\sigma)\frac{\mathscr{V}}{L(\epsilon_k^2)}$, where  
  $\mathscr{V}>0$ is the same constant mentioned in the proof of 
  Lemma~\ref{lem:V}. 

	Next, note that, $\hat{u}(X,t)$ is a continuously differentiable 
	function of time. Thus, by using Lemma~\ref{lem:V} and similar 
	arguments as before, we can say that $\norm{\dot{\hat{u}}(X,t)}$ is 
	upper bounded by a monotonically increasing positive real valued function of $\epsilon_k^2$, $\forall t \in [t_k,t_{k+1})$, $\forall k \in \nat$.
 
  Then, $\forall t \in [t_k, \min \{t_k+T,t_{k+1}\})$, $\forall k \in 
  \nat$,
	\begin{equation*}
	\frac{d}{dt}\norm{e(t)} \le \norm{\dot{e}(t)} \le \norm{\dot{u}(t)}+\norm{\dot{\hat{u}}(X,t)}\le \alpha_e(\epsilon_k^2),
	\end{equation*}
	for some monotonically increasing function $\alpha_e:
        \real_{> 0} \to \real_{> 0}$. The last inequality follows from
        Lemma~\ref{lem:claim}. Thus, for each $k \in \nat$,
        $t_{k+1}-t_k$ is lower bounded by
        $(1-\sigma)\frac{\mathscr{V}}{\alpha_e(\epsilon_k^2)L(\epsilon_k^2)}>0$. Now, from Lemma~\ref{lem:V}, we know that $\epsilon^2_1 \geq 
\epsilon^2_k$, for $k \in \nat$, where $\epsilon_1^2 := V(X(t_1),t_1) \ge \epsilon^2$. Hence, we can say that the inter-event 
times, for $k \in \nat$, are uniformly lower bounded by 
$(1-\sigma)\frac{\mathscr{V}}{\alpha_e(\epsilon_1^2)L(\epsilon_1^2)}>0$. This completes 
the proof of the first statement of this theorem.

Next note that, according to Lemma~\ref{lem:V} and the fact that the 
sequence of IETs does not exhibit Zeno behavior, the 
event-triggering rule~\eqref{eq:etr} implies that $\dot{V}(X(t),t) < 
-\sigma \Sigma(X(t),t) < 0$ for all $t\ge t_0$ such that $V(X(t),t) \ge 
\epsilon^2$. Thus, there exists $\bar{t} \in [t_0, \infty)$ such that 
$V(X(t),t) \le \epsilon^2$ for all $t \ge \bar{t}$. Under Assumption~\ref{A:A1}, this implies that the tracking error $X$ is uniformly ultimately bounded. Moreover, we can also 
say that there exists a finite $\bar{k}\in \nat$ such that 
$\epsilon_k^2=\epsilon^2$ for all $k \ge \bar{k}$. Thus, the lower 
bound on the IETs converges to 
$(1-\sigma)\frac{\mathscr{V}}{\alpha_e(\epsilon^2)L(\epsilon^2)}>0$ in 
finite time and events. This completes the proof of the last two statements of this theorem.
\end{proof}

\section{Simulation and Experimental Results}\label{sec:results}

In this section, we present results of our simulations and
experiments of the proposed ETPC for trajectory tracking. We compare the proposed method with 
time-triggered control (TTC) and the zero-order hold based 
ETC algorithm described in~\cite{RP-etal:2015}. 
We present the 
simulation and experimental results for four reference trajectories 
that were generated using the unicycle 
model~\eqref{eq:reference_trajectory}. The resulting paths in these 
four cases are shown in Figure~\ref{fig:ref-traj}.
\begin{figure}[]
  \begin{subfigure}[b]{.49\columnwidth}
    \includegraphics[width=.8\linewidth]{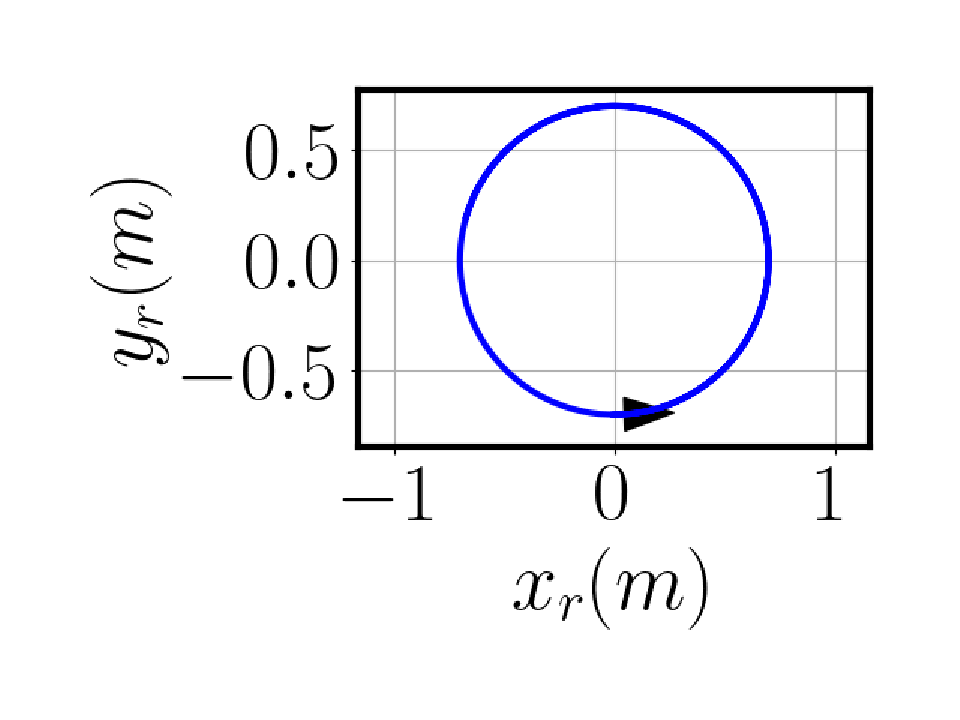}
    \label{fig:p1-ref}
  \end{subfigure}
  \hfill
  \begin{subfigure}[b]{.49\columnwidth}
    \includegraphics[width=.8\linewidth]{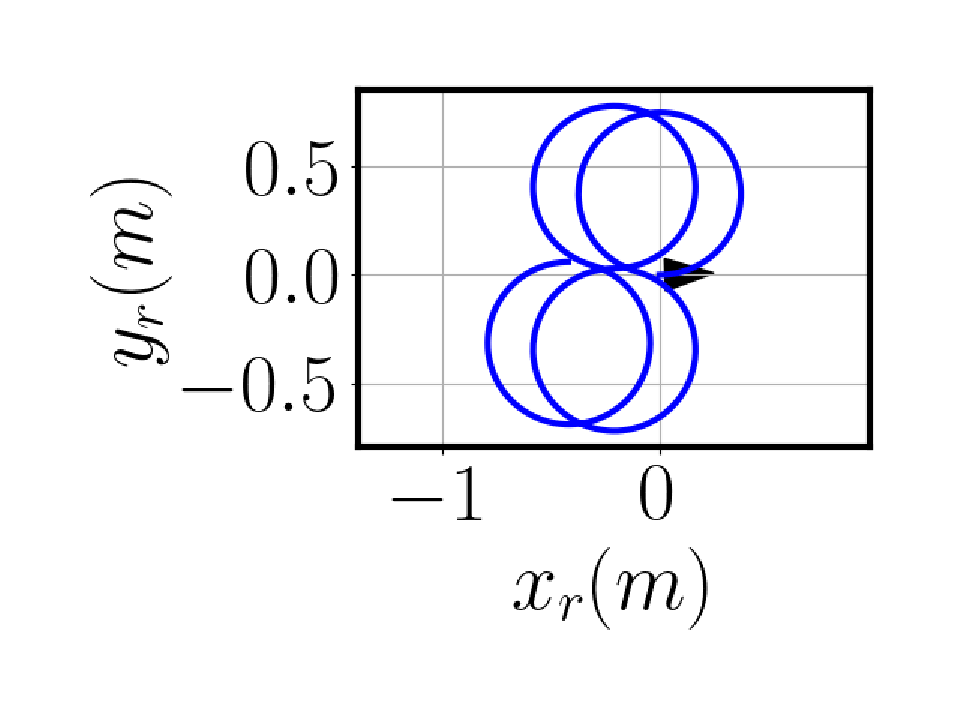}
    \label{fig:p2-ref}
  \end{subfigure}
  \begin{subfigure}[b]{.49\columnwidth}
    \includegraphics[width=.8\linewidth]{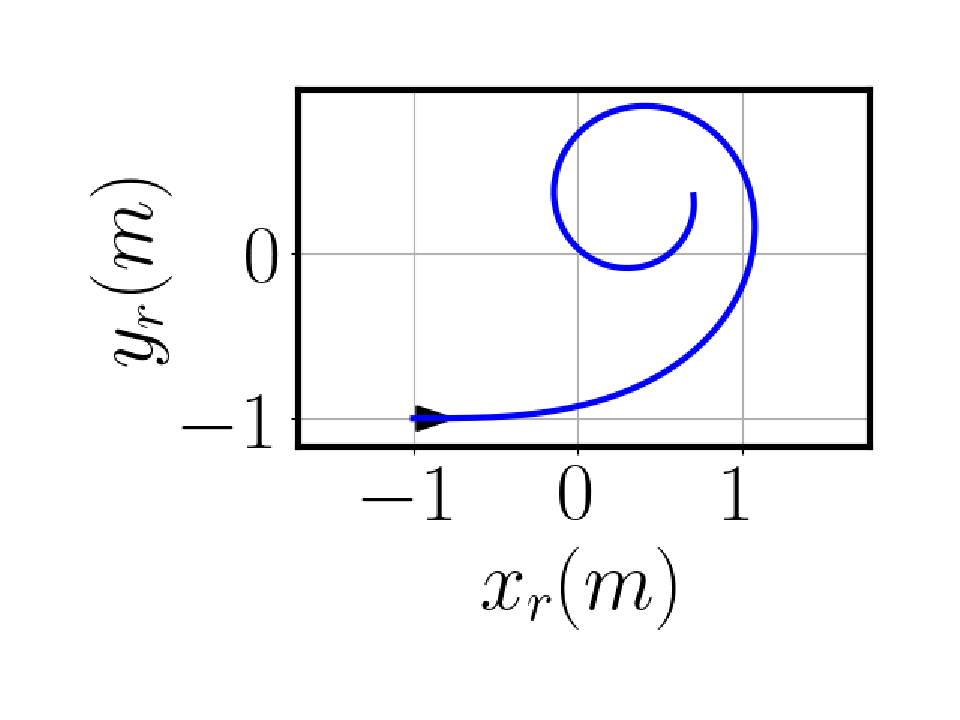}
    \label{fig:p3-ref}
  \end{subfigure}
  \hfill
  \begin{subfigure}[b]{.49\columnwidth}
    \includegraphics[width=.8\linewidth]{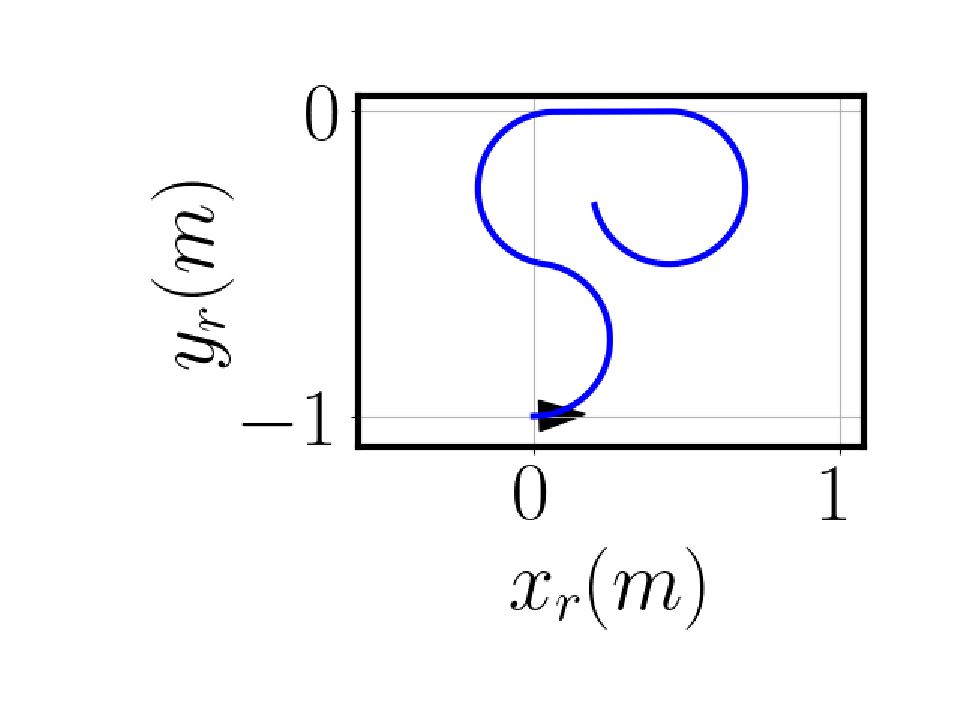}
    \label{fig:p4-ref}
  \end{subfigure}
  
  \caption{Reference trajectories under consideration a) Path $1$ , b) Path $2$,
c) Path $3$, d) Path $4$}
  \label{fig:ref-traj}
\end{figure}
The reference velocity for these trajectories is $15$cm/s. In paths 1
and 3, the reference angular velocity, $\omega_r$ is constant and
smoothly changing, respectively, while in paths 2 and 4, $\omega_r$ is
piecewise constant.

\subsubsection*{Evaluation metrics}

Let $T_e$ be the total simulation/experiment time duration. We define 
the transient period, (TP) as the time interval $[0,T_{c}]$, where
\begin{equation*}
  T_{c} \coloneqq \min \{t \geq 0 \ \colon V(t) \leq \epsilon^2\},
\end{equation*}
that is $T_c$ is the time the Lyapnuov  function $V$ takes to first
hit a value below $\epsilon^2$. We define the steady state period, (SS) as the 
time interval $(T_{c},T_e]$. Let $N_t$ and $N_s$ be the total number of 
events in TP and SS, respectively.

We compare the proposed ETPC with ETC and TTC in terms of the number of 
events, $N_t$ and $N_s$, and the convergence time $T_c$. In order to do 
a fair comparison with TTC, we 
compare ETC and ETPC against TTC with two different transmission 
frequencies/periods. In particular, henceforth, TTC$1$ and TTC$2$ refer 
to TTC with average transmission frequency over $T_e$ for ETC and ETPC 
respectively, in the corresponding simulation or experiment.

In general, the UU bound of $V(t)$ for TTC1 and TTC2 is higher 
than $\epsilon^2$. In practical experiments with ETC and ETPC also the 
UU bound is often higher than $\epsilon^2$. This is 
due to several 
unmodeled features, including sampling rate for the motion capture 
system (restricted to 240 frames per second here), measurement latency, 
error in obtained pose information, computation times for solving the 
optimization problem~\eqref{eq:a_k}, communication delays and latency, 
delay introduced by onboard serial communication on robot, 
environmental conditions such as slip and non-uniform surface friction. 
Even the kinematics of the robot may not exactly be unicycle model. 
Given all this, another evaluation metric we employ is $\bar{\epsilon}^2 
\geq \epsilon^2$, the UU bound for $V(t)$. Specifically, we 
define $\bar{\epsilon}^2$ as
\begin{equation*}
    \bar{\epsilon}^2 \coloneqq \max_{t \geq T_{c}} \{ V(t) \}.
\end{equation*}

\subsection{Simulations}
\label{sec:simulation}

Simulations were done for the system~\eqref{eq:full-system} with four 
reference trajectories that generate the paths shown in 
Figure~\ref{fig:ref-traj}. In the simulations, the integration time 
step was a fixed value of $5$ms for all the algorithms. The design parameters were chosen as $\gamma = 100$, $c_1 = 0.02$, 
$c_2 = 0.05$, $c_3 = 0.01$, $\sigma = 0.5$ and $\epsilon = 0.1$. The prediction time horizon for ETPC, $T$ was chosen as $1$ second and polynomial degree was chosen as $3$. The 
initial pose error was sampled uniformly from the set $[(-2$m, $2$m$)$, 
$(-2$m, $2$m$)$, $(-0.2$radians, $0.2$radians$)]$ to get a set of 
$1000$ initial conditions. Simulations were conducted for each 
algorithm, for each of the four paths in Figure~\ref{fig:ref-traj}, for 
each of these $1000$ initial conditions.

Figure~\ref{fig:N-sim} depicts the number of events for paths $1$ to 
$4$ in the transient and steady state period for ETC and ETPC.
\begin{figure}[!htb]

\begin{subfigure}[b]{0.47\columnwidth}
    \includegraphics[scale = 0.26, center]{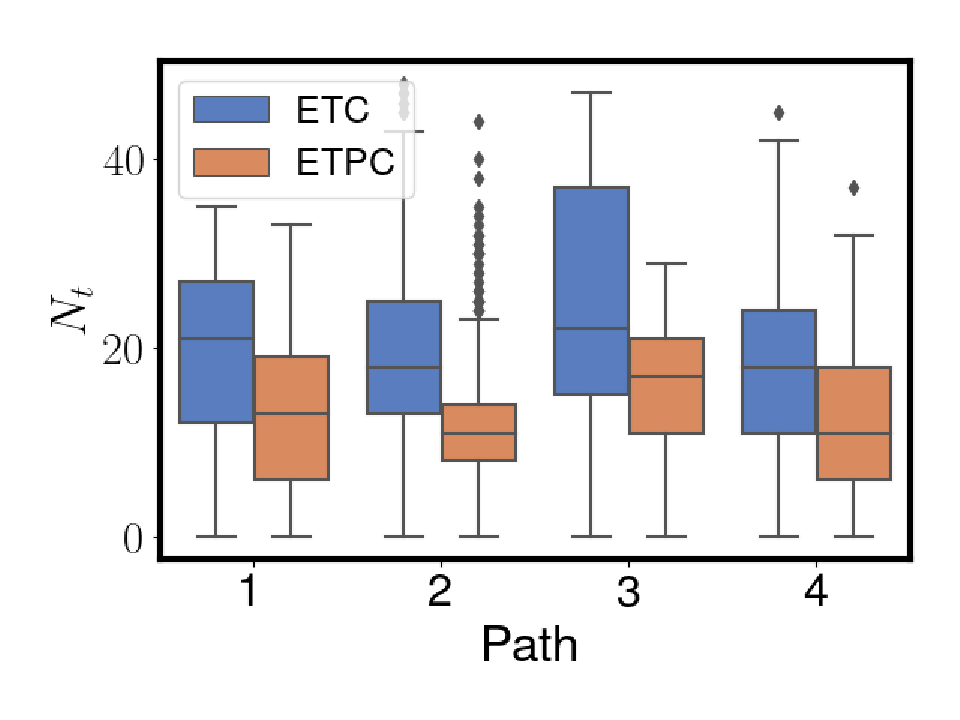}
    \caption{Number of events in TP}
  \end{subfigure}
  \hfill  
  \begin{subfigure}[b]{0.47\columnwidth}
    \includegraphics[scale = 0.26, center]{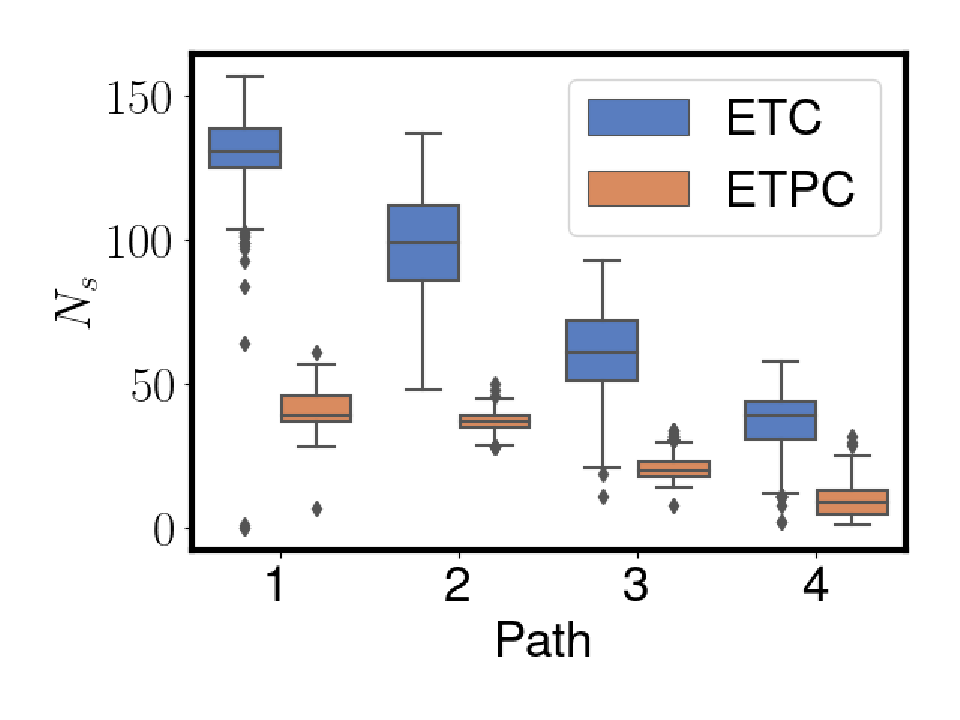}
    \caption{Number of events in SS}
  \end{subfigure}  
  \caption{Comparison of number of events for simulated paths for 
    algorithms under consideration.}
  \label{fig:N-sim}
\end{figure}
It is observed that the median of $N_t$ for ETPC is reduced by $38.1$\%, $38.9$\%,
$22.8$\% and $38.9$\% in comparison to median of $N_t$ for ETC for
paths $1$ to $4$, respectively. Similarly, the median of $N_s$ for
ETPC shown in Figure~\ref{fig:N-sim} is reduced by $70.3$\%, $62.7$\%,
$67.3$\% and $77$\% in comparison to median of $N_s$ for ETC for paths
$1$ to $4$, respectively. 
The third quartile of $N_s$ for all paths is much lower with ETPC than
even the first quartile with ETC. Thus, we can say that our algorithm
requires far fewer number of events than ETC in steady state and also,
to a lesser extent, during the transient period.

Figure~\ref{fig:bounds-sim} shows the UU bound of $V$ for TTC$1$ and
TTC$2$ algorithms for all the paths. The UU bound of $V$ for ETC and
ETPC is within numerical tolerance of $\epsilon^2$ bound. TTC1 and
TTC2 algorithms often give a UU bound in the order of hundreds. Thus,
we can conclude that the tracking performance with ETC or ETPC is
significantly better than that of TTC with comparable frequency of
transmissions.
\begin{figure}[h]
  \centering
  \includegraphics[width=0.2\textwidth]{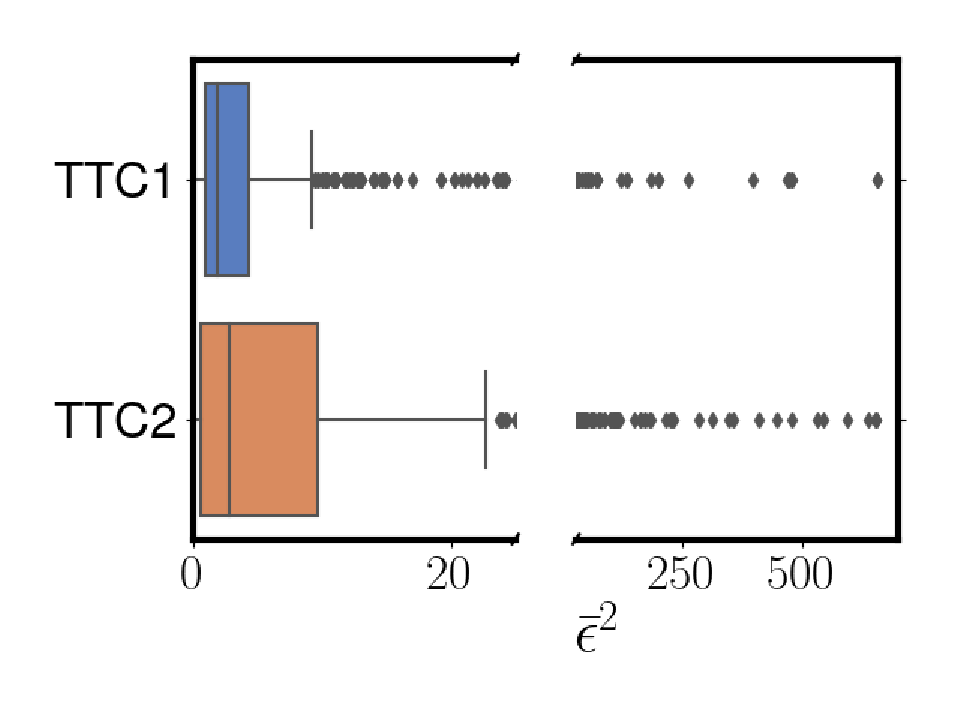}
  \caption{UU bound for all simulated paths.}
  \label{fig:bounds-sim}
\end{figure}

Figures~\ref{fig:traced-path-and-lyapunov} illustrates the result of one simulation for Path~3 with the initial 
pose error $(-1.02$m $,1.08$m $,0.142$rad$)$. Figure~\ref{fig:p3-sim} 
shows the reference path and the paths traced by ETC and ETPC while 
Figure~\ref{fig:lyp-sim} shows the evolution of Lyapunov function. 
It is observed that the error reduces over time and then oscillates to
ensure that $V$ stays within $\epsilon^2$ bound. It is observed that
the Lyapunov function under both ETC and ETPC stays within the
$\epsilon^2$ bound once it enters it.
 Notice that in Figure~\ref{fig:etpc-sim-error},
$\theta_{e}$ oscillates significantly even at the end of the 
simulation. However the corresponding $V$ stays within $\epsilon^2$ 
bound once it enters it as seen in Figure ~\ref{fig:lyp-sim}. This 
happens because the contribution of $\theta_{e}$ to the Lyapunov 
function $V$ in~\eqref{eq:Lyap_fn} is significantly low with $1/(2 
\gamma) = 1/200$.
\begin{figure}[!htb]
\begin{subfigure}[t]{0.45\columnwidth}
\centering
    \includegraphics[scale = 0.26]{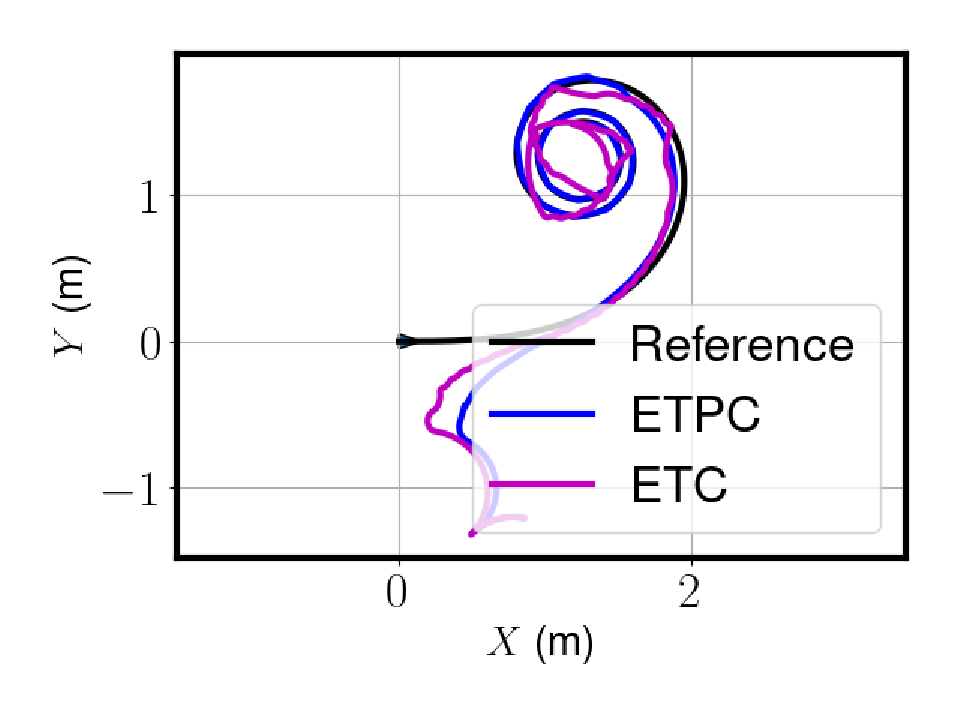}
    \caption{Path traced by robot}
    \label{fig:p3-sim}
  \end{subfigure}
\hfill
\begin{subfigure}[t]{0.45\columnwidth}
    \centering
    \includegraphics[scale = 0.25]{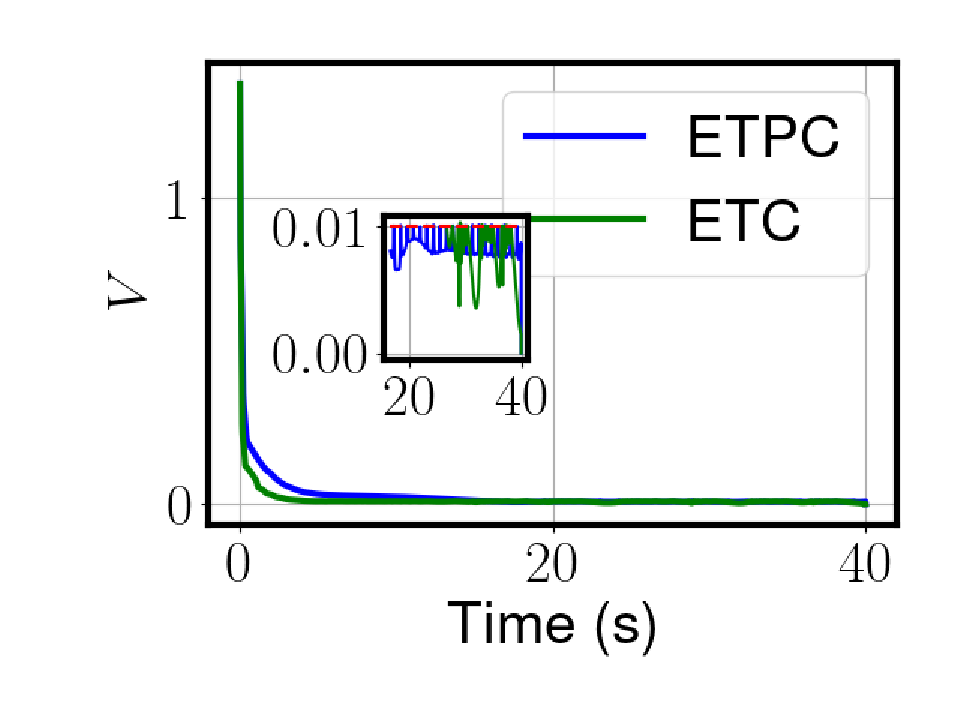}
    \caption{Evolution of $V$}
     \label{fig:lyp-sim}
  \end{subfigure}
\caption{Results of a simulation of robot tracking the reference 
trajectory that generates Path 3.}
\label{fig:traced-path-and-lyapunov}
\end{figure}
\begin{figure}[!htb]
  \centering
  \includegraphics[width=0.4\textwidth]{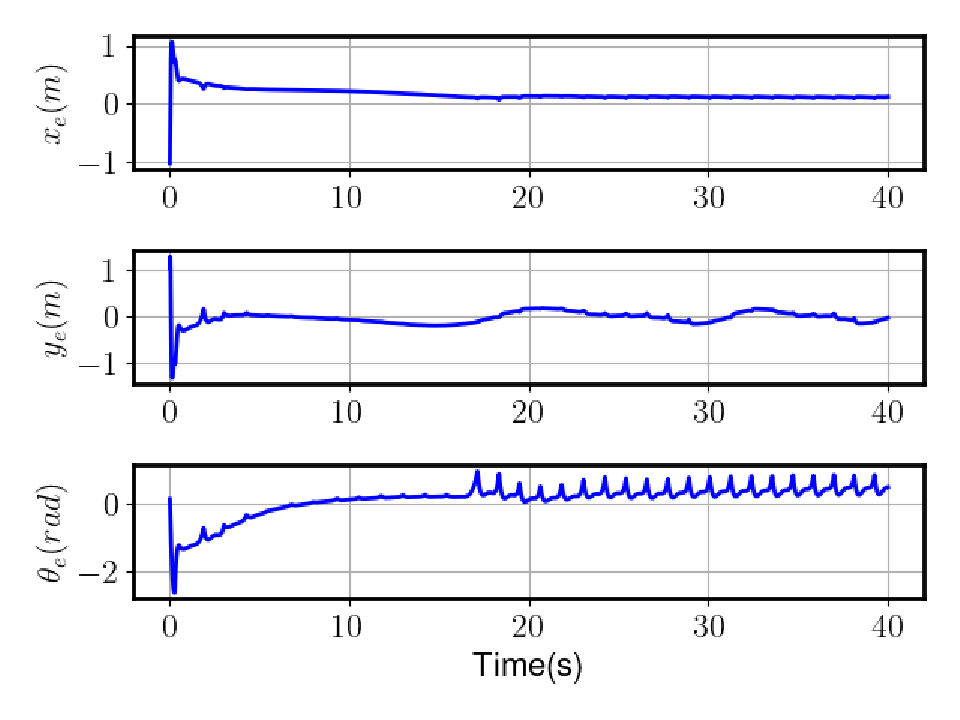}
  \caption{Evolution of pose error for simulated Path $3$ with initial 
  pose error of $(-1.02$ m$,1.08$ m$,0.142$rad$)$ with ETPC.}
  \label{fig:etpc-sim-error}
\end{figure}

\subsection{Practical Experiments}

Experiments have been conducted on a 3pi+ 32U4 robot 
manufactured by Pololu Robotics and Electronics
The robot is equipped with two quadrature encoders, which are utilized 
for robot level wheel velocity control done using a conventional PID 
controller. The microcontroller on the robot is an AtMega32U4 with 
$28$kB of memory available to the user. An RNXV WiFly module is 
interfaced with the robot using appropriate electronics  for wireless 
communication with a desktop computer. The 
desktop computer has a 64-bit Windows $11$ operating system with 
installed RAM of 40 GB and an Intel i7-8700 CPU with clock speed of 
3.20 GHz. The desktop computer is also interfaced with OptiTrack motion 
capture system \cite{optitrack}, which provides the pose measurements 
of the robot, with mean error in position of $5.4$ mm. The sampling 
rate of the motion capture system is 240 frames per second, which 
translates to a sampling period of about $4.1$ms. The motion capture 
system communicates the measurements over a wired single hop network 
connection and the latency in this communication could be up to $1$ms. 
The desktop computer monitors the event-triggering rule and also 
computes the coefficients of ETPC whenever required. This computational 
latency on the computer depends on the prediction horizon $T$ and 
integration time step and could go up to $25$ms. After an event occurs 
and the desktop computer computes the ETPC coefficients, it 
communicates them to the robot wirelessly. The sum of latency in the 
wireless network, control loops on the robot and the latency caused by 
onboard serial communication on the robot is on average between $15$ms 
to $20$ms but sometimes could go up to $100$ ms. 

The design parameters for experiments were chosen so that ETC and ETPC 
have a similar $\bar{\epsilon}^2$ UU bound on $V$. 
In particular, the 
chosen parameters are $\gamma = 1$, $c_1 = 0.5$, $c_2 = 0.8$, $c_3 = 
0.7$, and $\sigma = 0.9$. All quantities to be communicated are restricted to a precision of two decimal places. In our experiments, we use Transmission 
Control Protocol (TCP) which is an upper layer protocol to the Internet 
Protocol(IP). The size of each IP packet is fixed at $64$ Bytes, in 
which $46$ Bytes are reserved for data~\cite{rfc894}. For ETC, TTC$1$ 
and TTC$2$, the size of payload (actual intended message) is $4$ Bytes 
while it is $16$ Bytes for ETPC. A packet for tranmission is 
constructed by padding the actual payload to acheive the packet sizes 
required by the protocol. Thus, in comparison to ETC, TTC$1$, and 
TTC$2$, ETPC communicates a lot more information in a single packet 
less often without affecting communication overheads in each packet.

For each path shown in Figure~\ref{fig:ref-traj}, two initial 
pose errors were considered and for each initial pose 
error and each algorithm, $10$ experiments were conducted. The specific 
initial pose errors, in (m,m,radians), were chosen as $(0,0,0)$, 
$(1,0.2,1.57)$ for Path $1$;  $(0,0,0)$, $(0.2,0.2,1.57)$ for Path 
$2$; $(0,0,0)$, $(0,-0.8,0)$ for Path $3$; and $(0,0,0)$, 
$(-0.05,-0.5,-0.02)$ for Path $4$.

Fig. ~\ref{fig:p3-exp} shows the path 
traced by ETC, ETPC and 
TTC2 for Path $3$. It is seen that the robot eventually starts tracking 
the path with ETC and ETPC algorithms as seen by the corresponding 
decrease in Lyapunov function as shown in Fig.
~\ref{fig:P3-V} with $V$ eventually staying within a small uniform 
ultimate bound. For TTC$2$, we see that the UU 
bound on $V$ is much higher than for ETC and ETPC and the tracking 
behaviour is not good. 
Note that due to the unmodeled 
high computational and communication latency as well as other 
disturbances and modeling errors, the UU bound of 
$V$ for ETC and ETPC are also higher than the designed $\epsilon^2$.
\begin{figure}[!htb]
  \begin{subfigure}[t]{0.45\columnwidth}
    \centering
    \includegraphics[scale = 0.26]{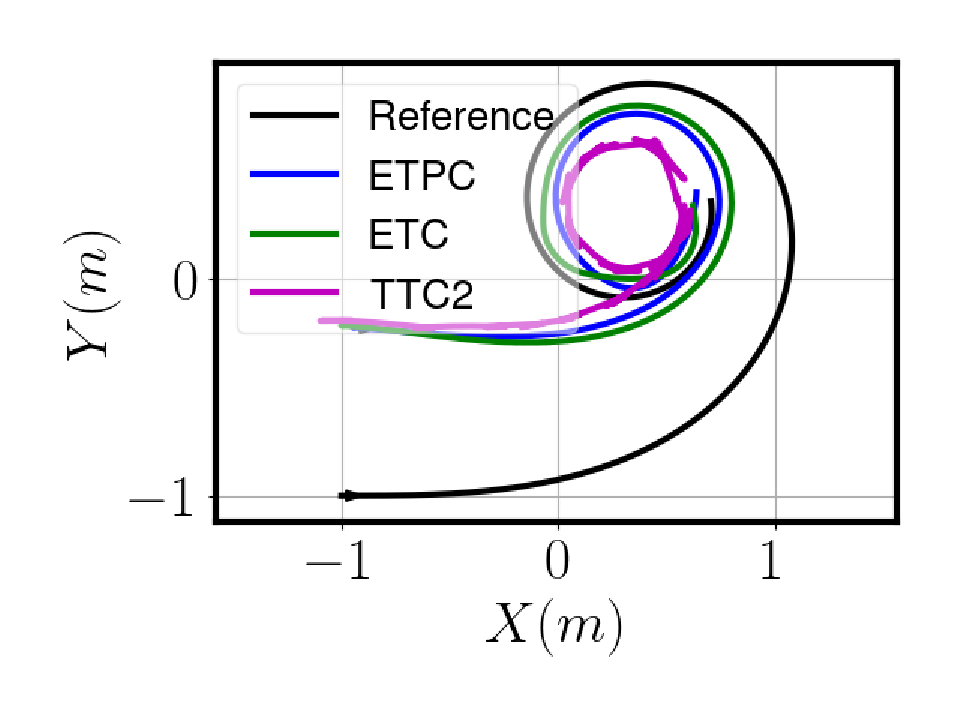}
    \caption{Path traced by robot}
    \label{fig:p3-exp}
  \end{subfigure}
  \hfill
  \begin{subfigure}[t]{0.45\columnwidth}
    \centering
    \includegraphics[scale = 0.25]{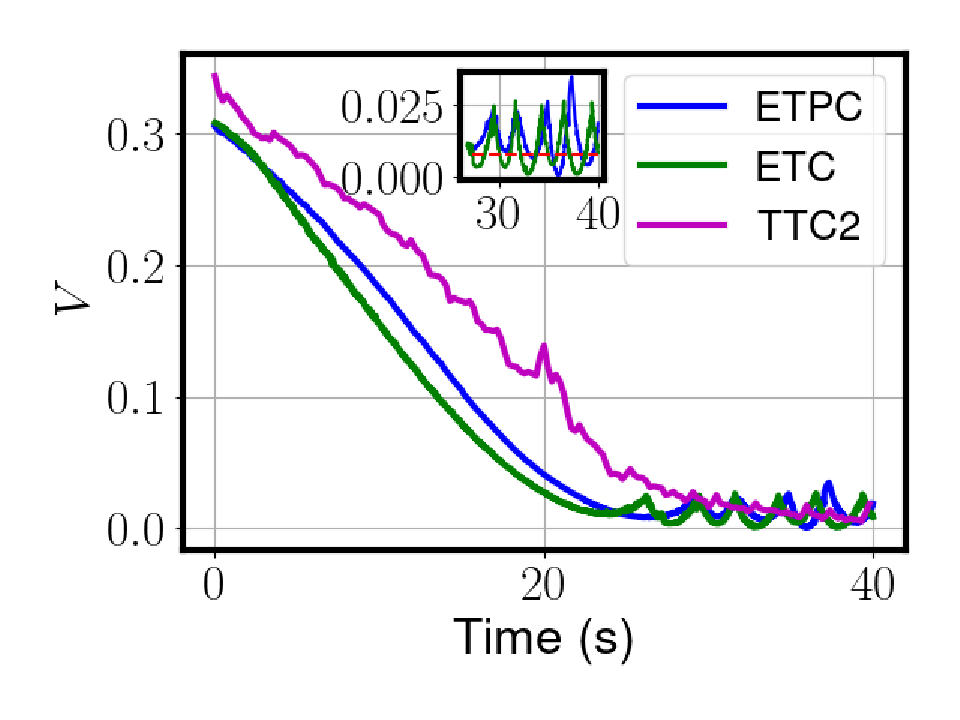}
    \caption{Evolution of $V$}
    \label{fig:P3-V}
  \end{subfigure}
  \caption{Results of an experiment of robot tracking the reference 
    trajectory that generates Path 3.}
\end{figure}

In transient period, the median of $N_t$ for ETPC is reduced by
$80.4$\% in comparison to the median of $N_t$ for ETC as seen in
Figure~\ref{fig:Nt-exp}. Similarly, Figure~\ref{fig:Ns-exp} shows the
number of steady state events for all paths. It is observed that the
median of $N_s$ for ETPC is reduced by $45$\% in comparison to median
of $N_s$ for ETC. Thus, we conclude that in both transient and steady
state, our algorithm has fewer number of
events. Figure~\ref{fig:tc-exp} shows the convergence time to
$\epsilon^2$ bound for all algorithms. We see that, in all cases, ETC
and ETPC have lower convergence times in comparison to TTC even when
outliers are considered. The convergence times for ETC are slightly
better than those of ETPC. We also observe that the UU bound as seen
in Figure~\ref{fig:eps1-exp} is much lower for ETC and ETPC than for
TTC and somewhat similar for ETC and ETPC.
From the suite of simulations and experiments, we can
conclude that our algorithm has fewer number of events than ETC while
ensuring similar tracking behaviour. Compared to TTC with a similar
average transmission frequency, the performance of ETC and ETPC are
far superior.

\begin{figure}[h]
\begin{subfigure}[b]{.49\columnwidth}
    \includegraphics[width=0.8\linewidth]{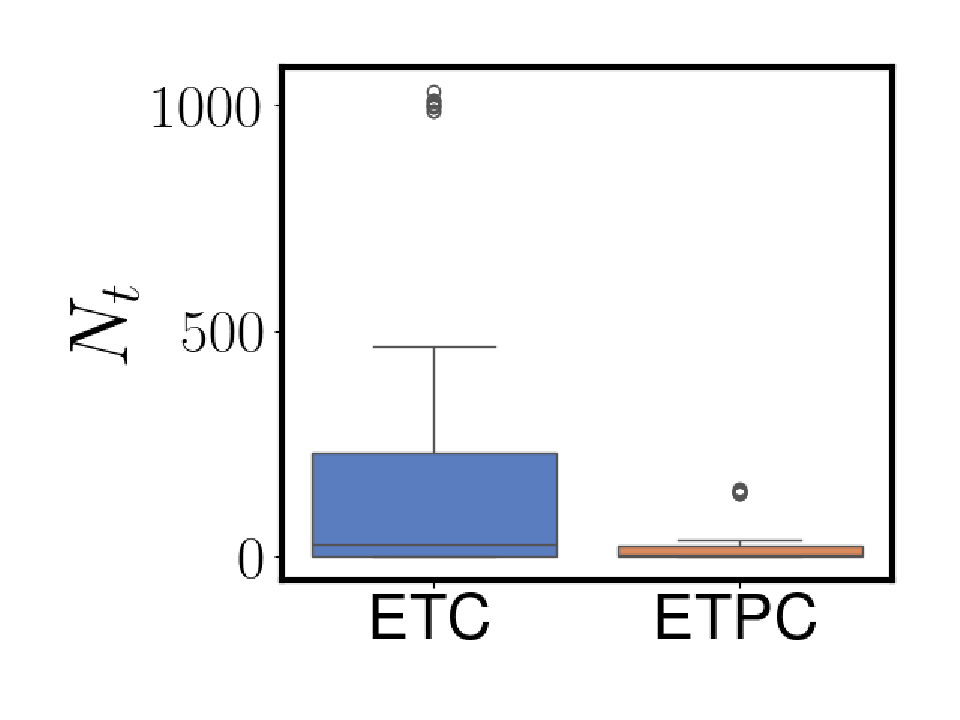}
    \caption{Number of events in TP}
     \label{fig:Nt-exp}
  \end{subfigure}
\begin{subfigure}[b]{.49\columnwidth}
    \includegraphics[width=0.8\linewidth]{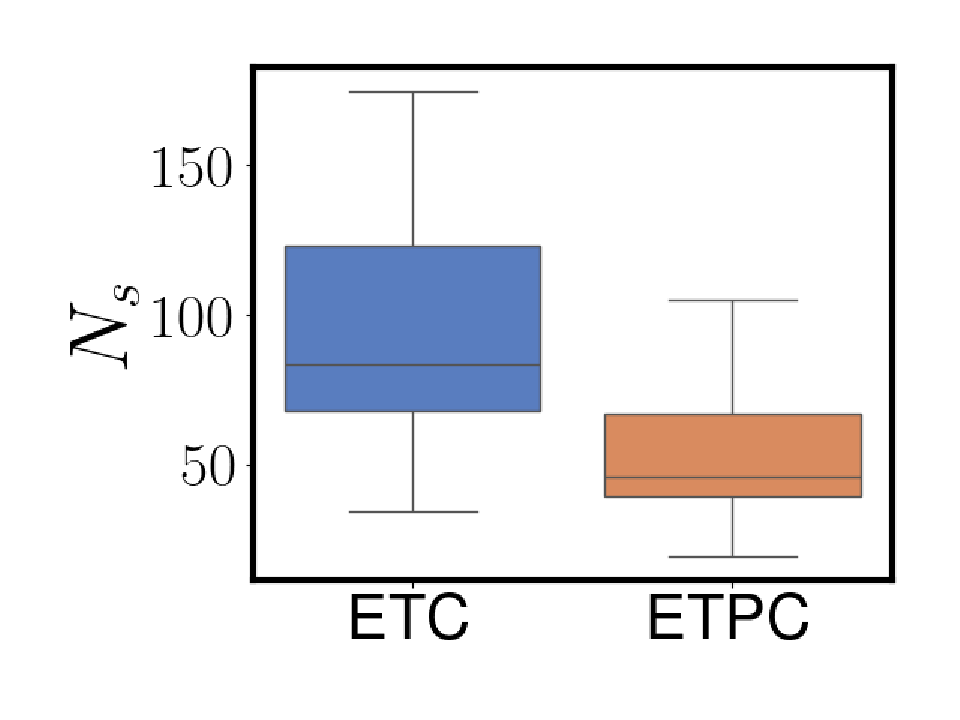}
    \caption{Number of events in SS}
     \label{fig:Ns-exp}
  \end{subfigure}
\begin{subfigure}[b]{.49\columnwidth}
    \includegraphics[width=0.8\linewidth]{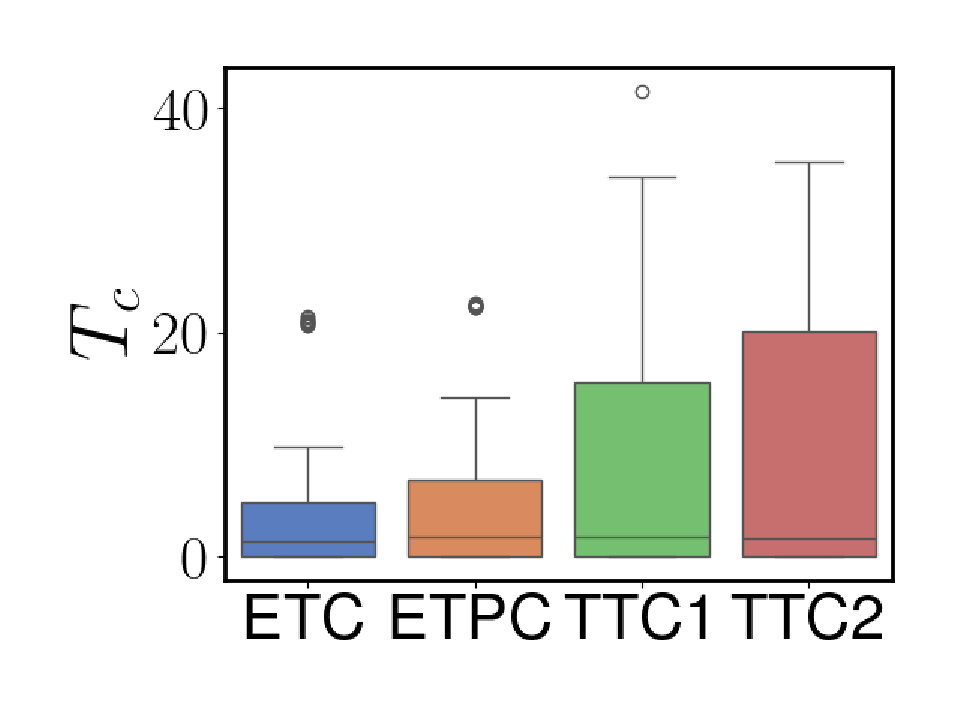}
    \caption{Convergence to $\epsilon^2$ bound}
     \label{fig:tc-exp}
  \end{subfigure}
\begin{subfigure}[b]{.49\columnwidth}
    \includegraphics[width=0.8\linewidth]{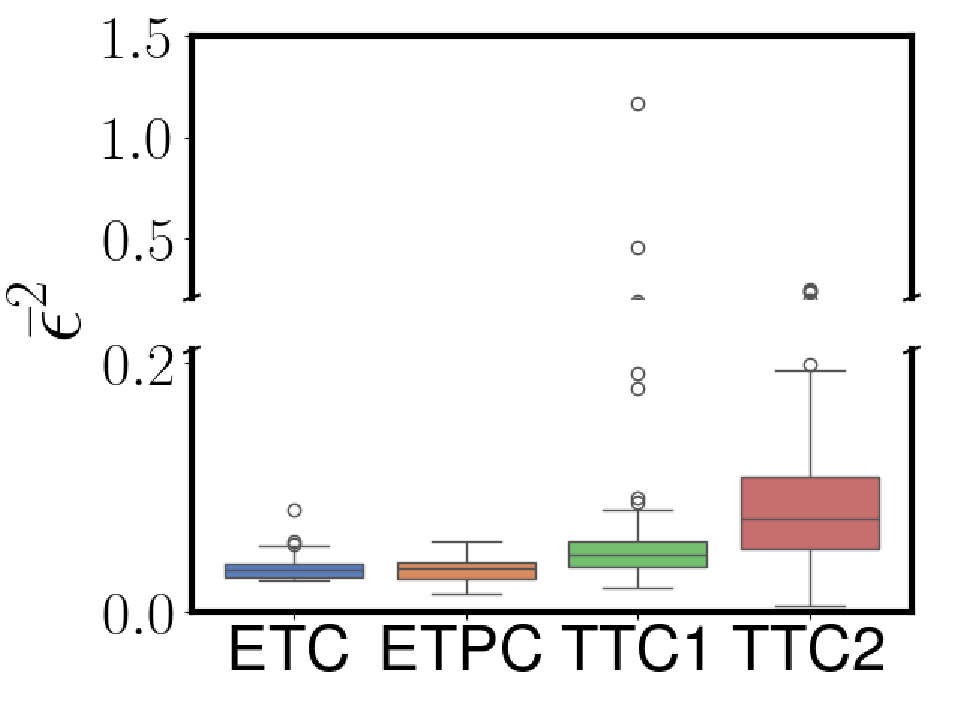}
    \caption{UU bound of V}
     \label{fig:eps1-exp}
  \end{subfigure}

\caption{Results of practical experiments for all the 
paths and all algorithms.}
\label{fig:exp-metrics}
\end{figure}

\section{CONCLUSION}\label{sec:conclusion}
In this paper, we proposed an ETPC method 
for trajectory tracking for unicyle model of robots where the reference 
trajectory is modeled as the solution of a reference unicycle model. 
We designed an 
event-triggering rule that guarantees UU boundedness of 
tracking error and non-Zeno behavior of IETs. The proposed 
method works best in cases where communication is significantly more 
costly than computation. We illustrated the results through numerical 
simulations and experiments. We also showed that the number of events 
generated by the proposed controller is significantly less compared to 
a time-triggered controller and an event-triggered controller based on 
zero-order hold, while guaranteeing similar tracking performance. 
Future work includes control under input disturbances, time delays, 
quantization of the parameters and multi-robot control.

\bibliographystyle{IEEEtran}
\bibliography{references}

\end{document}